\documentclass[10pt]{article} 
\usepackage{hyperref}
\usepackage{natbib}

\usepackage{color,graphicx}
\usepackage[english]{babel} 
\usepackage{t1enc}
\usepackage{times}
\usepackage{amsmath,amssymb,latexsym,amsthm}

\usepackage[width=15.5cm,a4paper,height=23cm]{geometry}
\newtheorem{theorem}{Theorem}

\newtheorem{coroll}{Corollary}
\newtheorem{lemma}{Lemma}
\theoremstyle{definition}
\newtheorem{definition}{Definition}
\theoremstyle{remark}

\definecolor{subtler}{rgb}{1,0,0.1}  
  
\newcommand{\be}{\begin{equation}} \newcommand{\ee}{\end{equation}}
\newcommand{\ba}{\begin{eqnarray}} \newcommand{\ea}{\end{eqnarray}}
\def\E{{\hbox{I\kern-.2em\hbox{E}}}} 

\newcommand{\Z}{\mathbf{Z}} \newcommand{\R}{\mathbf{R}}
\newcommand{\N}{\mathbf{N}}
\def\bl{\Bigl(\,}
\def\br{\,\Bigr)}

\def\sumj{\sum_{j=-\infty}^\infty}
\def\sumjj{\sum_{j=1}^\infty}
\def\sumk{\sum_{k=-\infty}^\infty}
\def\inth{\int_{-1/2}^{1/2}}
\def\BV{\mathrm{V}}         
\def\gstar{\gamma_{H}^*}
\def\cf{c_{\!f}}
\def\cffgn{\cf^*}
\def\cffar{\cf}

\def\ngoinf{\stackrel{n\rightarrow\infty}{\rightarrow}}
\def\ngoinfsim{\stackrel{n\rightarrow\infty}{\sim}}

\def\xgozerosim{\stackrel{x\rightarrow0}{\sim}}

\newcommand{\textdef}[1]{\textit{#1}}
\newcommand{\eref}[1]{(\ref{#1})}

\newcommand{\vx}{\ensuremath{\mathcal{V}}}
\newcommand{\vxm}{\ensuremath{\m{\mathcal{\vx}}}}
\newcommand{\g}{\gamma}
\newcommand{\w}{\ensuremath{\omega}}
\def\p{\phi}     

\newcommand{\m}[1]{\ensuremath{#1^{(m)}}}

\newcommand{\I}{\ensuremath{\mathbf{I}}}
\def\whv{\omega^*_{H,\vx}}

\usepackage{xspace}

\graphicspath{{.}
{Figures/}
}

\def\oneup{0.77}

\begin{document}

\date{\today}

\title{Why FARIMA Models are Brittle}
\author{A. Gorst-Rasmussen$^{\mathrm{a,b}}$, D. Veitch$^{\mathrm{a}}$, A. Gefferth$^{\mathrm{a,c}}$}
\maketitle
\vspace{-5mm}
$^\mathrm{a}$\textit{Department of Electrical and Electronic Engineering, The University of Melbourne, Victoria 3010, Australia.}

$^\mathrm{b}$\textit{Current affiliation: Department of Mathematical Sciences, Aalborg University, 9220 Aalborg East, Denmark.} 

$^\mathrm{c}$\textit{Current affiliation: Innovaci\'o i Recerca Industrial i Sostenible, Avda.~Carl Friedrich Gauss n${}^o$11, 08860 Castelldefels, Spain.}

\textbf{Email adddresses:}  \ \verb|agorstras@gmail.com, dveitch@unimelb.edu.au, andras.gefferth@gmail.com|

\begin{abstract}
  The FARIMA models, which have long-range-dependence (LRD), are
  widely used in many areas. Through deriving a precise
  characterisation of the spectrum, autocovariance function, and
  variance time function, we show that this family is very atypical
  among LRD processes, being extremely close to the fractional
  Gaussian noise in a precise sense. Furthermore, we show that this
  closeness property is not robust to additive noise.  We argue that
  the use of FARIMA, and more generally fractionally differenced time
  series, should be reassessed in some contexts, in particular when
  convergence rate under rescaling is important and noise is expected.
\end{abstract}


\section{Introduction}

For a wide variety of purposes including data modelling, synthetic
data generation, and the testing of statistical estimators, tractable
and flexible time series models are indispensible.  The well known
\textit{AutoRegressive Moving Average} (ARMA) family, for example,
allows for a wide variety of short range correlation structures, and
has been used in many contexts.

Long-Range Dependence (LRD), or long memory, in stationary time series
is a phenomenon of great importance \cite{DOT_Taqqu1}.  The
\textit{Fractional AutoRegressive Integrated Moving Average} (FARIMA)
models \cite{Hosking81,GrangerJoyeux1980} are very widely used as a
class which inherits the advantages of ARMA, while exhibiting LRD with
tunable \textit{Hurst parameter}, the scaling parameter of LRD.  They
have in particular been widely used to parsimoniously model data sets
exhibiting LRD (for example \cite{forecastingFARIMA_2000}),
and more importantly for our purposes
here, they have also been employed to make quantitative assessments of
the behaviour of stochastic systems in the face of LRD (for example \cite{Barbe_2010}).

A good example is in relation to estimators of the Hurst parameter $H$.  
FARIMA models have been used (for example
\cite{Comparestimate95,TaqquTeverovsky97,doukhan}) 
in order to evaluate the performance of $H$ estimators under
circumstances more challenging than that of the canonical
\textit{fractional Gaussian Noise} (fGn), in particular to assess 
small sample size performance using Monte Carlo simulation.  
Although explicit claims of the generality of the FARIMA family are not made, 
 implicitly it is taken to be a typical class of LRD time series in some sense, and so results obtained using it are taken to be representative for LRD inputs in general.

 In fact, no parametric model can be truly typical. However, for a model class to be useful it should be representative for the purposes to which it is commonly put.
In this paper, we show that FARIMA time series, and more generally
time series whose LRD scaling derives directly from fractional
differencing such as the FEXP models \cite{Robinson94}, are far from typical when it comes to their LRD character,  the very quality for which they were first introduced.
In a sense we make precise, out of all possible LRD time
series, their LRD behaviour is in fact `as close as possible' to that
of fGn.  A key technical consequence is ultra-rapid convergence to fGn
under the rescaling operation of aggregation.  The implications for
the role of the family is strong, namely that, in regards to LRD
behaviour, \textit{FARIMA offers no meaningful diversity beyond fGn}. 
A second key consequence is that the addition of additive noise  (of almost any kind) pushes
 a FARIMA process out of the immediate neighbourhood of fGn,
changing the convergence rate. In other words FARIMA is structurally
unstable in this sense or \textit{brittle},  and is therefore unsuited for use as a 
class of LRD time series  representing real-world signals.  

This work arose out of our prior study of (second-order)
self-similarity of stationary time-series \cite{discrete}, which
highlighted the benefits of the variance time function (VTF)
formulation of the autocovariance structure, over the more commonly
used  autocovariance function (ACVF) formulation. Using the
VTF, questions of process convergence under rescaling to exactly
(second-order) self-similar limits can often be more simply stated and studied.  

The paper is structured as follows.  
After Section~\ref{sec:back} on background material,
Section~\ref{sec:VTF} establishes the main results.  It begins by 
characterising a link between a fractionally differenced process and fGn
in the spectral domain. Using it, we prove that related Fourier
coefficients in the time domain decay extremely quickly, and then show
that as a result the VTFs of the fractionally differenced process and
fGn are extremely close.   
 We then explain why this behaviour is so atypical, and how it results in fast convergence to fGn.
Finally we go on to provide distinct direct proofs of closely related results for the ACVF and spectral formulation
 which are of independent interest. 
In particular, they lead to additional closeness results for the spectrum.  
In Section~\ref{sec:brittle} we explain why fractional processes are not robust to the addition of additive noise, even noise of particularly non-intrusive character.  We also provide numerical illustrations of this brittleness, and of the fast convergence to fGn of FARIMA processes.  
We conclude and discuss possible implications of our findings in Section~\ref{sec:conc}.

Very early versions of this work appear in the 2002 workshop papers \cite{hsn2002s,hsn2002a}.

\section{Background}
\label{sec:back}

Let $\{X(t), \,t\in\Z\}$ denote a discrete time second-order
stationary stochastic process.  The mean $\mu$ and variance $\vx>0$ of
such a process are independent of $t$, and the \textdef{autocovariance
  function} (ACVF), $\g(k):=E[(X(t)-\mu)(X(t+k)-\mu)]$, depends only on
the lag $k$, $k\in\Z$, and $\g(k)=\g(-k)$.

A description of the autovariance structure which is entirely
equivalent to $\g$ is the variance time function, defined as $\w(n) =
(\I\gamma)(n) := \sum_{k=0}^{n-1} \sum_{i=-k}^k \g(i)$
$n=1$, $2$, $3,\ldots$, where $\I$ denotes the double integration
operator acting on sequences.  Its normalised form, the
\textdef{correlation time function} (CTF), is just $\p(n)=\w(n)/\w(1) = \w(n)/\vx$.  
In terms of the original process, $\w(n)$ is just the variance of the sum
$\sum_{t=1}^n X(t)$.  It is convenient to symmetrically extend $\w$
and $\p$ to $\Z$ by setting $\w(n):=\w(-n)$ for $n<0$ and $\w(0)=0$.


\subsection{LRD, Second-Order Self-Similarity, and Comparing to fGn}
\label{ssec:lrdss}

There are a number of definitions of long-range dependence, all of which encapsulate the idea of slow 
decay of correlations over time.  
Common definitions include power-law tail decay of the ACVF 
$\gamma(n) \!\ngoinfsim \!c_\gamma n^{2H-2}$,  
or power-law divergence of the spectral density at the origin 
$f(x) \xgozerosim \cf |x|^{-(2H-1)}$  for related constants $c_\g$ and $\cf$ (see for example \cite{DOT_Taqqu1}, Section~4).

The well known \textdef{fractional Gaussian noise} (fGn) family, parameterised by the 
\textit{Hurst parameter} $H\in[0,1]$ and variance $\vx>0$, has $\w(m)=\whv(m) := \vx m^{2H}$ (to lighten notation we sometimes write $\w^*_H$ or simply $\w^*$).
It has long memory if and only if $H\in(1/2,1]$.

In this paper we compare against fGn with $H\in(1/2,1]$ as it plays a
special role among among LRD processes; that of being a family of
\textit{second-order self-similar} time series\footnote{Until
  recently, fGn was considered to be the only such family. A second
  (and final) family was discovered recently \cite{newfps}.}.  To
understand how this comparison can be made, we must define
self-similarity and related notions.

Self-similarity relates to invariance with respect to a rescaling operation.  
In the present context, the time rescaling is provided by what is commonly called
\textdef{aggregation}.  For a fixed $m \ge 1$, the \textit{aggregation
of level $m$} of the original process $X$ is the process $\m X$ defined as
\begin{displaymath}
	\m X(t) :=\frac 1 m \sum_{j=m(t-1)+1}^{mt} X(j).
\end{displaymath}
The $\g$, $\w$, $\p$ functions and the variance of the $m$-aggregated
process will be denoted by $\m \g$, $\m \w$, $\m \p$ and $\vxm$ respectively.
It is not difficult to show \cite{discrete} that 
\begin{equation}
	\m \w(n)=\frac{\w(mn)}{m^2}  ,\qquad \vx^{(m)} = \frac{\w(m)}{m^2}.
       \label{eq:wm}
\end{equation} 

To seek invariance, the time rescaling must be accompanied by a
compensating amplitude rescaling.  This is performed naturally by
dividing by $\vx^{(m)}$, which amounts to examining the effect of
aggregation on the correlation structure.  Combining the time and
amplitude rescalings yields the correlation renormalisation
\be
	\m \p(n) = \frac{\p(mn)}{\p(m)}  = \frac{\w(mn)}{\w(m)}. 
       \label{eq:phim}
\ee
We can now define second-order self-similarity as the fixed points of this operator.
\begin{definition}[]
\label{def:SS}
A process is second-order self-similar iff $\m \p = \p$, for all $m=1,2,3,\ldots$.
\end{definition}
It is easy to see that fGn, which has $\p(m)=\p_H^*(m) := m^{2H}$, satisfies this definition for all $H\in[0,1]$.


Given a fixed point $\p_H^*(n)$, we define its \textdef{domain of
  attraction} (DoA) to be those time series which converge to it
pointwise under the action of \eref{eq:phim}.  This definition is very
general, in particular it includes processes whose VTF's have
divergent slowly varying prefactors, as these cancel following
normalization (see Section~\ref{ssec:atyp}).
It provides a natural way to define LRD which subsumes and generalises
most other definitions including those above \cite{discrete}:
\textit{a time series is long-range dependent if and only if it is in
  the domain of attraction of $\p_H^*(n)$ for some $H \in (0.5,1]$}.

With the above definitions the DoA are revealed as the natural way to
partition the space of all LRD processes, namely into sets of
processes each corresponding to the same unique normalized fGn fixed
point.  Since all processes within a DoA converge to the same fixed
point, their asymptotic structure can be meaningfully compared both
against each other and to the fixed point itself.  Alternatively if
two processes were in different DoA's then they cannot be close
asymptotically as they would converge to different processes.  
  Section~\ref{ssec:closeVTF} provides a precise characterisation of
  the closeness of a fractionally differenced process to its
  corresponding fixed point, and its associated fast convergence under
  renormalization.  

Within a given DoA, one can further partition
processes according to some measure of distance from the common fixed point. 
Section~\ref{ssec:atyp} establishes such a notion, enabling a comparison
of this closeness to that of other members of the DoA to be made.

\subsection{Fractionally Differenced Processes and FARIMA}
\label{ssec:farima}

Let $B$ denote the backshift operator. The fractional differencing operator of order $d>-1$ is given by 
  \begin{displaymath}
      (1-B)^{d} := \sum_{j=0}^\infty  \Gamma(j-d)/\Gamma(-d)\Gamma(j+1) B^j .   
  \end{displaymath}
Let $\{Y(t), t \in \Z\}$ be a second-order stationary stochastic process.   Assuming $H\in(0,1)$ the process 
  \begin{displaymath}
      X := (1-B)^{-(H-1/2)}Y    
  \end{displaymath}
is called a fractionally differenced process with differencing parameter $H-1/2$ driven by $Y$.

If $h$ is the spectral density of $Y$ then $X$ has spectral density (\cite{tstm}, Thm.~4.10.1)
\begin{equation}
    \label{eq:farimaeq}
    f_H(x) = h(x) \big|1-\mathrm{e}^{2\pi ix}\big|^{-(2H-1)} = h(x) |2\sin \pi x|^{-(2H-1)},  \quad x \in [-1/2,1/2].
\end{equation} 
 
In this paper we assume that $Y$ is short-range dependent, and in particular that $h$ satisfies:
\vspace{-2mm}
\begin{list}{$\bullet$}{\setlength\itemsep{-0.1cm}}
   \item $h(x)>0$ and is continuous for all $x \in [-1/2,1/2]$ (and is therefore bounded);
   \item $h$ is three times continuously differentiable on $(-1/2,1/2)$ (and is therefore in $C^3$).
\end{list}
\vspace{-2mm}
Under such conditions, the ACVF of $X$ exists and satisfies
$\gamma_H(n) \thicksim c_\gamma n^{2H-2}$ for some constant $c_\gamma$
(\cite{tstm}, Thm.~13.2.2). 
Hence, when $H\in(1/2,1)$ the process $X$ is LRD with Hurst parameter $H$. 

An important example of a fractionally differenced process is the
 FARIMA class \cite{Hosking81}  where $h$ is the spectral density
of a causal invertible ARMA model.  This family includes the ARMA family as the special case $H=1/2$.
Another class is the class of FEXP-models (e.g.~\cite{Bloomfield73, Robinson94,Beran93}) which comes from taking the
logarithm of $h$ to be a trigonometric polynomial, 
i.e.~$\log h(x)=\theta_1 \cos x +\theta_2 \cos(2x)+ \cdots + \theta_{q-1}\cos((q-1)x)$  for real coefficients. 
Both FARIMA and FEXP models are widely used in statistical applications since, in addition to
exhibiting LRD, they both enable modelling of arbitrary short-range correlation structures.

\subsection{ Normalizing a Fractionally Differenced Process to its fGn Limit}
\label{ssec:spectra}

To identify the fGn fixed point of a fractionally differenced time
series only the value of $H$ need be determined.  When aggregating an
unnormalised fractionally differenced time series however, to identify
the corresponding limiting fGn time series we must in additional know the
correct variance $\vx$.  The purpose of this section is to define
notation to make this simple and along the way to provide useful
expressions for the spectra of these processes.

The ACVF, VTF, and spectral density corresponding to the fixed point are denoted
$\gamma_H^*$, $\omega_H^*$, and $f_H^*$, respectively. The latter is given by 
(see \cite{Taqqu})
\begin{eqnarray}
    \label{eqn:ffgn}
    f^*_{H}(x) &=& \cffgn \,\pi^{-2} (2\pi)^{2H+1}\sin^2(\pi x)  \sumj  |2\pi j  +  2\pi x|^{-(2H+1)}\\\nonumber
                     &\xgozerosim & \, \cffgn |x|^{-(2H-1)},  \quad x \in [-1/2,1/2] ,
\end{eqnarray}
where $\cffgn = \vx (2\pi)^{2-2H} C(H) >0$ is the prefactor of the power-law at the origin, and 
$C(H) =\pi^{-1} H \Gamma(2H)\sin(H\pi)$ (see \cite{Taqqu}, pp 333-4, but note that the change to normalised frequency multiplies $f^*_H$ by $2\pi$, and $c_f^*$ by $(2\pi)^{2-2H} $).

We denote by $\gamma_H$, $\omega_H$ and $f_H$ the ACVF, VTF
and spectral density of a fractional process with Hurst parameter $H\in(1/2,1)$. 
In view of \eqref{eq:farimaeq}, the latter is given by
\begin{eqnarray}
    \label{eqn:ffar}
    f_{H}(x) = h(x) |2\sin \pi x|^{-(2H-1)}  \hspace{-3mm}&=& \cffar(2\pi)^{2H-1}  \frac{h(x)}{h(0)} |2\sin \pi x|^{-(2H-1)}\\\nonumber
                                                               &\xgozerosim & \cffar |x|^{-(2H-1)},  \quad x \in [-1\!/2,1\!/2] ,
\end{eqnarray}
where $\cffar = (2\pi)^{1-2H} h(0)>0$.
In the case of a pure fractionally differenced process, such as FARIMA($0,d,0$), $h(x)=h(0)=2\pi$, and $\cf=(2\pi)^{2-2H}$
(note again the changes related to normalised frequency, in particular the factor of $2\pi$ is built into $h(0)$).

To conclude, the particular fGn to which the fractionally differenced process will converge under renormalisation is the one such that $\cffgn=\cffar$.
>From this, the value of $\vx$ can be obtained using the expressions for $\cffgn$ and $\cffar$ above, if needed.

\subsection{Regularity and Other Notations}
\label{ssec:reg}

\smallskip
Denote for $\alpha \geq 0$ by $\Lambda_\alpha$ the normed space of uniformly $\alpha$-H\"older continuous
functions on $[-1/2,1/2]$,
\begin{displaymath}
    \Lambda_\alpha :=\big\{\varphi\colon [-1/2,1/2] \to \R\,:\,
    \|\varphi\|_{\Lambda_\alpha}<\infty\big\},
\end{displaymath}
where $\|\cdot\|_{\Lambda_\alpha}$ is the $\alpha$-H\"older norm
\begin{displaymath}
    \|\varphi\|_{\Lambda_\alpha} := \sup_{x,y \in [-1/2,1/2]}|\varphi(x)-\varphi(y)||x-y|^{-\alpha}.
\end{displaymath}
Hence $\Lambda_{\alpha} \supseteq \Lambda_\beta$ whenever $\alpha\leq\beta$.  
The space $\Lambda_\alpha$ is closed under pointwise
multiplication, addition, and composition with functions in
$\Lambda_1$.  In particular, the subset of $\Lambda_\alpha$ whose
members are bounded away from zero is closed under reciprocation
(i.e.~if $g\in\Lambda_\alpha$, and $g$ is bounded away from zero, then
so is $1/g$).  Observe that $\varphi \in \Lambda_1$ whenever
$\varphi'$ exists and is bounded.  Functions in $\Lambda_\alpha$ are
absolutely continuous.

The linear space of functions of bounded variation on $[-1/2,1/2]$, denoted $\BV$, is defined by
\begin{displaymath}
     \BV:=\big\{\varphi\colon [-1/2,1/2] \to \R\,:\,
    \|\varphi\|_\BV<\infty\big\}, 
\end{displaymath}
where $\|\cdot\|_\BV$ is the total variation norm
\begin{displaymath}
     \|\varphi\|_\BV: = \sup\Big\{\sum_{i=1}^{|P|}  |\varphi(x_{i})-\varphi(x_{i-1})|\,:\, P=\{x_0,x_1,\ldots x_n\} \textrm{ is a partition of }[-1/2,1/2]\Big\}.
\end{displaymath}
$\BV$ is also closed under pointwise multiplication and addition (\cite{Apostol74}, Thm.~6.9), 
and reciprocation of those functions in $V$ bounded away from zero (\cite{Apostol74}, Thm.~6.10). 
Any differentiable function with bounded derivative on $(-1/2,1/2)$ is of bounded variation 
on $[-1/2,1/2]$ (\cite{Apostol74}, Thm.~6.6).

We shall use the notation $\star$ for convolution of sequences. For sequences $a$ and $b$
\begin{displaymath}
    (a \star b)_n = \sum_{j=-\infty}^\infty a_j b_{n-j}, \quad n \in \Z. 
\end{displaymath}
The convolution is said to exist if the infinite sum converges for all $n$.
When needed for clarity, we also use $(a \star b)(n)$ to denote $(a \star b)_n$.

Throughout, by \textit{smooth function} we mean one in $C^\infty$.

\section{Fractionally Differenced Processes are Not Typical LRD Processes}
\label{sec:VTF}

The goal of this section is to establish our main results, rigorous characterisations of the closeness
of the asymptotic covariance structure of a fractionally differenced process to  that of fGn.

Our approach is simple and can be described as follows.  We begin in
the spectral domain where the relationship between the processes can
be simply stated through a function $g$ by defining 
\begin{equation}
\label{eq:firsttimeg}
       f_H(x) =  f^*_{H}(x) g(x).  
\end{equation}
The simple closed form of the spectra \eref{eqn:ffgn} and
\eref{eqn:ffar} allow $g$ to be explicitly written.  We study the
properties of $g$, obtaining a characterisation of the closeness of
the processes in the spectral domain
(Theorem~\ref{lem:spec-dens-ratio}).  This leads to a convolution
formulation $\gamma_H= \gamma_{H}^*\star G$ in the time domain, where
$G$ is the Fourier Series of $g$, and thereby to a similar
relationship for the VTFs, where the fast decay of the Fourier
coefficients can be used to characterise the closeness
(Theorem~\ref{thm:conv-of-farima-vtf-2}). The VTF result then allows
the closeness within the DoA and the convergence speed to be easily
established (Theorem~\ref{thm:convspeed}).  Finally we also provide
direct closeness results for the ACVF
(Theorem~\ref{prop:conv-of-farima}).

\subsection{Closeness of the Spectrum}
\label{ssec:closespec} 


We are ultimately interested in characterising the closeness of the
covariance structure of a fractionally differenced process to that of
its fGn fixed point at large lags. The rate of decay of the sequence
of Fourier coefficients of a function is well known to be closely
connected to its smoothness properties. It is, therefore, unsurprising
that a notion of closeness in the spectral domain can take the
form of statements about smoothness of the function $g$ in
\eref{eq:firsttimeg}.

  The following spectral closeness result is the crucial basis for
  both the VTF and ACVF results to come.

\newcounter{qcounter}
\begin{theorem}
\label{lem:spec-dens-ratio}
Assume that $H\in[1/2,1)$ and define $g(x):=f_H(x)/f_{H}^*(x)$, $x \neq 0$ and $g(0):=\lim_{x \to 0}g(x)$. 
Then $g(0)=\cffar/\cffgn = h(0)/(2\pi \vx C(H))$ and $g$ satisfies the following over $[-1/2,1/2]$\textrm{:}\vspace{-0.3cm} 
\begin{list}{\textup{(\roman{qcounter})}~}{\usecounter{qcounter}\setlength\itemsep{-0.1cm}}
\setlength\labelwidth{1cm}
\setlength\itemindent{0.3cm}
   \item $g$ is even, continuous, positive, bounded, and $L^p$, $p>0$;
   \item $g$ is twice differentiable, and smooth away from $x=0$;
   \item $g'' \in \Lambda_{2H-1} \cap \BV$, but $g''\not\in\Lambda_{\beta'}$ for $\beta'>2H-1$;
   \item $g$ admits a Fourier series with coefficients $\{G_j\}$ such that $\sumj j^2|G_j|< \infty$ and $G_n=O(n^{-3})$. 
            In particular $\sumj |G_j|<\infty$ and $\sumj j^\alpha|G_j|< \infty$ for $1<\alpha<2$.
\end{list}
\end{theorem}
\begin{proof}
Unless otherwise specified, we consider the domain $x\in [-1/2,1/2]$.\\
First, since $f_{H}(x) \xgozerosim \cffgn |x|^{-(2H-1)}$ and $f^*_{H}(x)\xgozerosim \cffar |x|^{-(2H-1)}$,   
$g(0):= \lim_{x \to 0} g(x) = \cffar/\cffgn$.

The proof of (i) is straightforward. For completeness, details are provided in the appendix.
%

To prove the smoothness properties (ii) and (iii), we first establish those of  $\tilde{g}$ defined as
\begin{eqnarray} 
   \label{eq:gtilde}
   \tilde{g}(x) &:=&  \frac{\cffar \pi^{2H+1}}{\cffgn h(0)} \cdot \frac{h(x)}{g(x)} \\
                     &=&  \Big|\frac{\sin(\pi x)}{\pi x}\Big|^{2H+1} +
     |\sin(\pi x)|^{2H+1}\sum_{\substack{j=-\infty \\ j \neq 0}}^{\infty}  |\pi j +  \pi x|^{-(2H+1)}  \\
                     &:=& |a(x)|^{2H+1} + |b(x)|^{2H+1} c(x)  .
    \label{eq:decom-of-f-et}
\end{eqnarray}  
%
It is not difficult to show (see the appendix for details) that $\tilde{g}$ is smooth everywhere 
except at the origin where its smoothness is controlled by that of $|b|^{2H+1}$, which we now study.

Let $\beta = 2H-1$. 
Since $b$ is smooth and $\beta\in(0,1)$, $|b|^{\beta+2}$ is twice differentiable at the origin.
The smoothness of its second derivative is controlled by $(b')^2|b|^{\beta}$, which,
since $b\in\Lambda_1$ and $x \mapsto |x|^\beta$ is in $\Lambda_{\beta}$, is also in $\Lambda_{\beta}$ 
by the multiplicative and compositional closure properties of $\Lambda_{\beta}$. 
It follows that $\tilde{g}''$ exists and is in $\Lambda_{\beta}$.
Since however $x \mapsto |x|^\beta$ is not in $\Lambda_{\beta'}$ for any $\beta'>\beta$, and 
moreover $b(x)\xgozerosim\pi x$ and $b'(0)\ne0$,  
$\tilde{g}''$ is not in $\Lambda_{\beta'}$ for any $\beta'>\beta$.

Since smooth functions are in $\BV$, by similar arguments using the closure properties of $\BV$,
we have ${\tilde g}''\in \BV$ if $|b|^{\beta} \in\BV$.  The latter holds since it is easy to see that 
$|b|^{\beta}$ is monotone (with total variation $2$).

We have shown that $\tilde{g}''$ exists and is in $\Lambda_{2H-1}\cap\BV$, 
but not in $\Lambda_{\beta'}$ for any $\beta'>2H-1$.
We now prove the same for $g$ using (\ref{eq:gtilde}). 
It suffices to consider $1/\tilde{g}$ since $h'''$ exists.
Since $\tilde{g}$ is bounded away from zero, (ii) follows since  $(1/\tilde{g})'' =  2(\tilde{g}')^2/\tilde{g}^3-\tilde{g}''/\tilde{g}^2$ clearly exists, and is smooth away from the origin. 
Now consider (iii).
It follows from the last expression and the fact that $\tilde{g}>0$  that $(1/\tilde{g})''$ and hence $g''$ are in $\BV$ and $\Lambda_{2H-1}$ by applying the respective closure properties.
Finally,  since $1/\tilde{g}^2(0)\ne 0$,  the smoothness of $(1/\tilde{g})''$ is controlled by that of $\tilde{g}''$
and so $(1/\tilde{g})''\not\in\Lambda_{\beta'}$ for any $\beta'>2H-1$. 
This completes the proof of (iii).

We now prove (iv).  
Since each of $g$, $g'$, and $g''$ are continuous and bounded, the
Fourier series for each exists and are related by term by term
differentiation (\cite{Champeney90}, Thm.~15.19).  In particular
$g(x)=\sumj G_j e^{2\pi ijx}$, and we can write $g''(x) = -4\pi^2
\sumj j^2 G_j e^{2\pi ijx}$.  Now Zygmund \cite{Zygmund_Vol1},
Thm.~VI.3.6 states that the Fourier Series of a function in
$\Lambda_{\beta}\cap\BV$ for some $\beta>0$ converges absolutely. This
applies to $g''$ and proves that $\sumj j^2|G_j|< \infty$ as claimed.
Finally, since $g''\in\BV$, the magnitude of its Fourier coefficients
decay as $O(|j|^{-1})$ (Zygmund \cite{Zygmund_Vol1}, Thm.~II.4.12),
proving that $G_j=O(j^{-3})$.
\end{proof}

The result suggests that fractionally
  differenced processes are not typical; for a general LRD process, only boundedness
  of $g$ at the origin would be automatic. In contrast, the present $g$ is a very well behaved function. 
A plot of $g$ is provided in Figure~\ref{fig:geg} which shows its flatness at the origin
(it also suggests that $g$ is monotone increasing over $[0,1/2]$, though this plays no role in what follows).  
Here we have set $\cffar=\cffgn$, so that its value at the origin is just $1$.
It is interesting to note that since $g$ is positive, even, and square
integrable, it is the spectral density of some second order time series.

\begin{figure}[h!]
\center
   \includegraphics[scale=\oneup,angle=270]{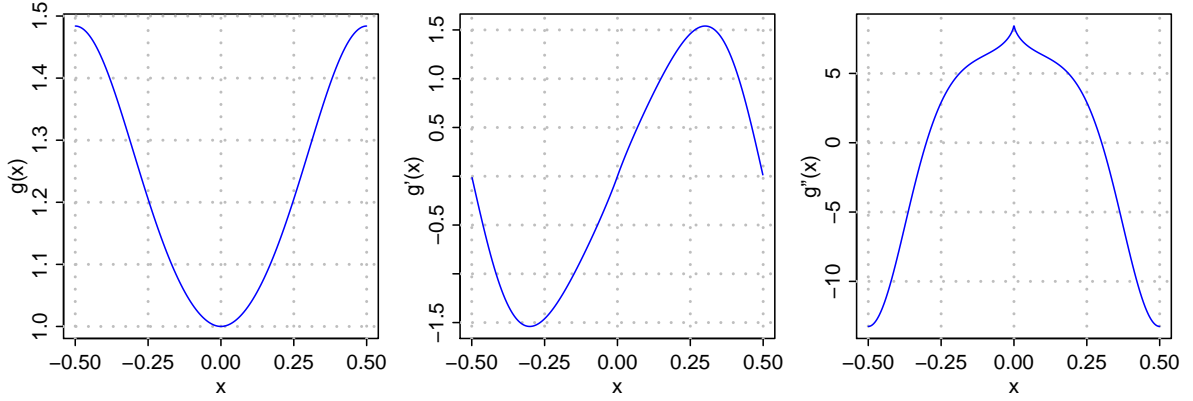}
\caption{\label{fig:geg} The function $g(x)=f_H(x)/f_H^*(x)$  and its first two derivatives in the canonical case of a pure fractionally differenced process (FARIMA0d0) with $H=0.8$ and $\cffar=\cffgn$.}
\vspace{-3mm}
\end{figure}

\subsection{Closeness of the VTF}
\label{ssec:closeVTF}

The first step in elucidating the relationship between $\w_H$ and
  $\w_H^*$ is to confirm that the relationship $f_H(x) = f^*_{H}(x)
g(x)$ between the spectral densities translates to the expected
convolution relationship $\gamma_H=\gstar \star G$ between the ACVFs.
It is straightforward to confirm that, thanks to the nice behaviour of
$g$ and $G$ detailed in Theorem~\ref{lem:spec-dens-ratio}, this is
indeed the case.
\begin{lemma}
   \label{lem:conv_exists}
   The auto-covariance functions $\gamma_H$ and $\gstar$  are related through the convolution $\g_H= \gstar \star G$.
\end{lemma}
\vspace{-2mm}
For completeness a proof is given in the appendix.

\medskip
Since $\w_H = \I\g_H$, it is tempting to seek a relationship of the form $\w_H= G\star \w_H^*$ through taking the `double integral' of $\gamma_H= G\star\gstar$.   However, since $\w_H^*(m)=\vx m^{2H}$ diverges with $m$, this is not necessarily well defined.  The following lemma provides a sufficient condition for the existence of such a convolution, as well as some of its important properties which will be crucial in what follows.
 
\begin{lemma}
\label{lem:main-lem-2}
Assume $1 <\alpha < 2$ and let $a=\{|n|^{\alpha}, n \in \Z \}$. 
Let $b$ be a symmetric sequence satisfying $\sum_{j=1}^\infty j^\alpha |b_j|<\infty$. 
Then $S_b=\sum_{j=-\infty}^\infty b_j$ and the symmetric sequence $c=a \star b$
exist, and $(c_n - S_b a_n)\ngoinf 0$.
\vspace{-2mm}
\end{lemma}
This result is proved in the appendix. The proof of the last part is based on the monotonicity of a function which generalises $\gstar$ to two parameters  (see 
Lemma~A1
in the appendix).

\begin{coroll}
\label{cor:Iconv_exists}
The convolution $G\star \w_H^*$ exists for $H\in(1/2,1)$.
\vspace{-2mm}
\end{coroll}
\begin{proof}
Set $b=G$ in Lemma~\ref{lem:main-lem-2}. 
The condition on $b$ holds since 
$\sum_{j=1}^\infty j^\alpha |G_j| < \sum_{j=1}^\infty j^2 |G_j|$ which is finite from Theorem~\ref{lem:spec-dens-ratio}.
The result then following immediately by identifying $\alpha$ with $2H$ and 
$a$ with $\omega_H^*$.
\end{proof}

 
The following lemma shows that, if existence is granted, taking the `double integral' of a convolution is straightforward, provided a double counting issue at the origin is allowed for.
\begin{lemma}
\label{lem:doub-int-of-conv}
  Let $a,b$ be symmetric sequences and assume that $c:=a \star b$ exists.  
  Then $\I c$ exists, and if $(\I a) \star b$ exists, then 
  $\I c = (\I a) \star b - \left((\I a) \star b\right)_{0}$ .
\vspace{-2mm}
\end{lemma}
The proof of this result is based on a careful rearrangement of terms justified by the repeated use of the existence of 
$(\I a) \star b$. It is given in the appendix.

\bigskip
We are now able to prove our main result on the VTF.
\begin{theorem}
 \label{thm:conv-of-farima-vtf-2}
  Let $\omega_H$ denote the VTF of a fractionally differenced process with $H\in(1/2,1)$ and $\cffar$ chosen equal to $\cffgn$. 
  Then
  \begin{displaymath}
      \omega_H(n) = \omega_{H}^*(n) + D + o(1)
  \end{displaymath}
  where $D = - 2\sum_{j=1}^\infty j^\alpha|G_j|<0$ is a constant.
\end{theorem}

\begin{proof}
Since each of $\gstar\star G$ and $\omega_H^* \star G$ exist, Lemma~\ref{lem:doub-int-of-conv} applies upon identifying $a=\gstar$, $b=G$ and $c=\gamma_H$ and states that 
$\omega_H=\omega_H^* \star G-\{\omega_H \star G\}(0)$.  
>From Lemma~\ref{lem:main-lem-2} with $b=G$, $S_G=\sumj G_j<\infty$ exists.  
By introducing the term $S_G\omega_H^*$ we obtain
\begin{eqnarray*}
  \omega_H  &=& S_G\omega_H^* + \bl \omega_H^* \star G - S_G \omega_H^* \br  - \{\omega_H \star G\}(0)\\
                     &=& S_G\omega_H^* +  o(1)                                                             - 2\sum_{j=1}^\infty j^\alpha|G_j| 
\end{eqnarray*}
by the final part of Lemma~\ref{lem:main-lem-2}.
Since $S_G=g(0)=\cffar/\cffgn=1$, the result follows.
\end{proof}

The key property underlying this result is $\omega_H^* \star G - S_G
\omega_H^* \ngoinf 0$, which shows that $G$ is `compact' enough to act
as an aggregate multiplier $S_G$ asympotically. This is analogous to
the role the covariance sum $S_\g :=\sum_{k=-\infty}^\infty\g(k)$
plays in the asymptotic variance of aggregated short-range dependence
processes \cite{discrete}.

\subsection{Atypicality and Speed of Convergence}
\label{ssec:atyp} 

Theorem~\ref{thm:conv-of-farima-vtf-2} showed that the VTF of a
fractionally differenced process is asymptotically equal to the VTF of
its fGn fixed point up to an additive constant.  This makes
fractionally differenced process highly atypical among LRD processes.
We show this first for the VTF itself, and then for the speed of
convergence of the CTF to the fixed point.

Without loss of generality, the VTF of any time series in the domain of attraction of a
given fGn can be expressed as 
\be
    \w_H(n) = \w_H^*(n) + \w_d(n)   
    \label{eq:wd}
\ee
where $\w_d$ represents the distance of the VTF from its
limiting fGn counterpart.  By definition, $\w_d(n)= o(n^{2H})$, but
otherwise the growth rate of $\w_d$  is not constrained,
implying that there is considerable variety within the domain of attraction.

One way of  characterising the size of the difference $\w_d(n)$ is to 
use \textit{regular variation} \cite{bgt,discrete}.  
 A regularly varying function $f(n)$ of index $\beta$ and integer argument $n\in\N^+$ satisfies
$\lim_{k\to\infty} f(kn)/f(k) = n^\beta$, $\beta\in\R$.
Assume without loss of generality that $\w_d$ is upper bounded by a regularly varying function of index $\beta\in[0,2H]$, that is
\be
    \w_d(n) = O( s(n)n^\beta ) ,
    \label{eq:RV}
    \ee where $s$ is a \textit{slowly varying} function (that is
    regularly varying with index $0$), and $\beta$ is the infimum of
    indices for which (\ref{eq:RV}) holds.  A notion of
    \textit{closeness} of the process to the limiting fGn can then be
    defined in terms of $\beta$, where the smaller the index, the
    closer the process.

According to this scheme, Theorem~\ref{thm:conv-of-farima-vtf-2}
states that fractionally differenced processes belong in the closest
layer of the hierarchy, corresponding to $\beta=0$.  Furthermore, the
theorem shows that $s(n)$ (which could in general diverge, for example
$s(n)\!\!\ngoinfsim\!\log(n)$) tends to a constant.  Thus, the VTF of
a fractionally differenced process lies in a very tight neighbourhood
indeed of the VTF of its limiting fixed point.  Far from being typical
LRD processes, they deviate only in very subtle ways from fGn in terms
of their large lag behaviour.

>From \eqref{eq:phim}, there is a direct relationship between
closeness in the above sense and speed of convergence of the CTF to
its fixed point under aggregation.

\begin{theorem}
 \label{thm:convspeed}
  Let $\p_H$ denote the CTF of a fractionally differenced process in the domain of attraction of $\p_H^*$ with $H\in(1/2,1)$. 
  Then 
  \begin{displaymath}
      \m\p_H(n) = \p_{H}^*(n) + D(1-n^{2H})m^{-2H} + o(m^{-2H}) = \p_{H}^*(n) + O(m^{-2H})
  \end{displaymath}
  where $D$ is the constant from Theorem~\ref{thm:conv-of-farima-vtf-2}.
\end{theorem}

\begin{proof}
The result follows from substituting $\omega_H(n) = \omega_{H}^*(n) + D + o(1)$ from 
Theorem~\ref{thm:conv-of-farima-vtf-2} in (\ref{eq:phim}) and using $(1+x)^{-1}=1-x + O(x^2)$.
\end{proof}  
  
Beginning from (\ref{eq:wd}), it holds generally for LRD processes in the DoA of $\p_H^*$
  that $\m\p_H(n) = \p_{H}^*(n) + O(s(m)m^{-2H+\beta})$.
It follows that fractionally differenced processes, for which $\beta=0$ and $s(m)$ is identically equal to a constant, 
converge faster to the fixed point compared to all other
processes in the DoA.  Examples are provided in Section~\ref{sec:brittle}.

\section{Closeness of the ACVF}
\label{sec:ACVF}

Recall that $\w=\I\g$.  Because the double sum operator $\I$ smooths out local variations,
Theorem~\ref{thm:conv-of-farima-vtf-2} can not be
used to derive an explicit characterisation of the closeness in terms
of the ACVF. 
We therefore set out to provide a closeness result for the ACVF here.  
Not only is this of interest in its own right, it also provides an alternative way of
demonstrating the closeness to fGn, as well as leading to an
additional result on the spectral closeness to fGn in an additive sense.

The following lemma is the analogue of Lemma~\ref{lem:main-lem-2} used for the ACVF.  
A proof is given in the appendix.
\begin{lemma}
\label{lem:main-lem}
Assume $-1 \leq \alpha < 0$ and let $a$ be the symmetric positive sequence $a_n=|n|^\alpha$, $n \neq 0$  and $a_0>0$. 
Let $b$ be a symmetric sequence with $|b_0|<\infty$ for which there exists  $\beta \in [0,2]$ such that 
$\sum_{j=1}^\infty j^\beta|b_j|<\infty$ and $|b_n| = O(n^{-(\beta+1)})$. 
Then $S_b:=\sum_{j=-\infty}^\infty b_j$ and the symmetric sequence $c:=a\star b$ exist, 
and $c_n- S_ba_n=O(n^{\alpha-\beta})$ as $n \to \infty$.
\end{lemma}

We can now prove the ACVF closeness result
\begin{theorem}
 \label{prop:conv-of-farima}
  Let $\g_H$ denote the ACVF of a fractionally differenced process with $H\in(1/2,1)$ and $\cffar$ chosen equal to $\cffgn$. 
  Then
  \begin{displaymath}
    \gamma_H(n) = \gamma_H^*(n) + O(n^{2H-4})
  \end{displaymath}
\end{theorem}
\begin{proof} 
  The exact ACVF of a unit variance fGn$(H)$ is given by
\begin{displaymath}
    \gamma_{H}^*(n) = \frac{1}{2}\big((n+1)^{2H}+(n-1)^{2H}-2n^{2H}\big),
\end{displaymath}
for $n \geq 0$ and $\gamma_H^*(n)=\gamma_H^*(-n)$ for $n<0$.  
Then $\gamma_{H}^*(0) = 1$, and for $n\ne0$ $\gamma_{H}^*(n)=(1/2)|n|^{2H}k(|n|^{-1})$ where $k(x):=(1+x)^{2H}+(1-x)^{2H}-2$. Expanding $k$ in a Taylor series around the origin, we obtain the following series representation:
\begin{displaymath}
   \label{eq:rep-of-acvf-fgn}
    \gamma_{H}^*(n) = \sum_{j=1}^\infty c_jf_j(n), \quad c_j:=\frac{\prod_{i=0}^{2j-1}(2H-i)}{(2j)!},
                                         \quad f_j(n) := \left\{\begin{array}{ll}
                                                                        |n|^{2H-2j}                      & n\ne0\\
                                                                        \mathbf{1}\{j=1\}/c_1   & n=0
                                                                     \end{array}   \right.
\end{displaymath}
which is uniformly absolutely convergent since $\{a_j\}$ is absolutely convergent by the ratio test. \\
Now $\gamma_H(n)= \big(\gstar \star G\big)(n) = \sumk G(k)\sumjj
c_jf_j(n-k)  = \sumjj c_j\big( f_j\star G\big)(n)$ where the existence of
$\gamma_H^* \star G$ and $\gamma_H^*$ as absolutely convergent series
justifies the interchange of summations (\cite{Apostol74}, Thm.~8.43).
We can now compare $\gamma_H$ and $\gstar$ as
\begin{eqnarray}
   \label{eq:est-for-diff}
   |\gamma_H(n) - \gamma_{H}^*(n) | &=& 
   \sumjj   |c_j| \big| (f_{\!j} \star G)(n) - f_j(n) \big|  \\
   &\leq& |c_1|\big|(f_1 \star G)(n)-f_1(n)\big|  + \sum_{j=2}^\infty |c_j|\big|(f_{\!j} \star G)(n)\big| + O(n^{2H-4});
\end{eqnarray}
We shall show that each of the terms on the right hand side are of order $O(n^{2H-4})$.\\ 
The result for the first term follows immediately from Lemma~\ref{lem:main-lem} upon identifying 
$f_1$ with $a$, $2H-2$ with $\alpha$, $G$ with $b$ with a choice of 
$\beta=2$ (justified by Theorem~\ref{lem:spec-dens-ratio}(iii)), 
and noting that $\sum_{j=-\infty}^\infty G_j=1$ by the assumption $\cffar=\cffgn$.

Now consider the second term.  Recall from
Theorem~\ref{lem:spec-dens-ratio} that $G_n=O(n^{-3})$, i.e.~there
exists $K>0$ such that $G_n\le K |n|^{-3}$ for $n$ sufficiently
large. Thus, when $j \geq 2$ and for $n>0$ large enough
\begin{eqnarray*}
 |(f_j \star G)(n)| &=&  \sumk |f_j(k)| |G_{n-k}| = \sum_{\substack{k=-\infty\\k \neq 0}}^\infty |k|^{2H-2j} |G_{n-k}| 
                                                                                   \le \sum_{\substack{k=-\infty\\k \neq 0}}^\infty |k|^{2H-4} |G_{n-k}| \\
   &=& \sum_{k=1}^\infty |k|^{2H-4} |G_{n+k}|  +   \sum_{k=1}^{\lfloor{n/2}\rfloor} |k|^{2H-4} |G_{n-k}|  + \!\sum_{k=\lfloor{n/2}\rfloor+1}^\infty \!|k|^{2H-4} |G_{n-k}|  \\
   &\le&   K\sum_{k=1}^\infty |k|^{2H-4} (n+k)^{-3}  +  K\!\sum_{k=1}^{\lfloor{n/2}\rfloor} |k|^{2H-4} (n-k)^{-3}  + \Big|\frac{n}{2}\Big|^{2H-4}\!\!\!\sum_{k=\lfloor{n/2}\rfloor+1}^\infty \!|G_{n-k}|  \\
  &\le&   K\sum_{k=1}^\infty (kn+k^2)^{2H-4}  +  K\!\sum_{k=1}^{\lfloor{n/2}\rfloor} (kn-k^2)^{2H-4} + \Big|\frac{n}{2}\Big|^{2H-4}\sumk \!|G_{n-k}|  \\
  &<&   K\sum_{k=1}^\infty (kn)^{2H-4}  +  K\!\sum_{k=1}^{\lfloor{n/2}\rfloor} (kn/2)^{2H-4} + \Big|\frac{n}{2}\Big|^{2H-4} \sumk \!|G_{k}| \\
 &=&  O(n^{2H-4});
\end{eqnarray*}
using $2H-4\ge-3$, that $kn+k^2 \geq kn$ for all $k$, $kn-k^2 \ge
nk/2$ for $1 \leq k \leq n/2$, the absolute summability of $G$, and
the fact that $\sum_{k=1}^\infty |k|^{2H-4}<\infty$.  Hence the right hand side of \eqref{eq:est-for-diff} is $O(n^{2H-4})$.
\end{proof}
In Section~\ref{ssec:closespec} we derived a result which may best be
described as `multiplicative closeness' for the spectrum of a
fractionally differenced process. This form of closeness was natural for
providing a subsequent link to the time domain. However, when
calculations with the frequency domain are of specific interest, an
additive closeness result for the spectrum is useful. Such a
result can easily be derived from the above theorem.
\begin{coroll}
\label{cor:spectral-closeness}
  It holds that $f_H(x) = f^*_{H}(x)+\varphi(x)$ where $\varphi$ is differentiable,
  $\varphi' \in \Lambda_{\alpha}$ if $\alpha<2-2H$, and
  $\varphi(0)=0$. Moreover,  $\varphi(x)=O(x^{-2H+3})$ as
  $x \to 0$.
\end{coroll}
\begin{proof}
  Let $\varphi:=f_H - f^*_{H}$.  The Fourier series of $\varphi$
  exists and equals $\varphi$, and its coefficients are given by
  $d_n=\gamma_H(n) - \gstar(n)$, which by Theorem~\ref{prop:conv-of-farima} is
  $O(|n|^{2H-4})$.  Since $2H-4<-2$ the first absolute moment of the
  coefficients exists, so Theorem~7.19 \cite{Kufner71}, applies and
  shows that $\varphi'$ exists and $\varphi' \in\Lambda_{2-2H}$.
  By the definition of $g(0)$, $\varphi(0)=\lim_{x \to 0}\big(f_H(x)-f_H(x)/f^*_{H}(x)\,f^*_{H}(x)\big)=0$.

  The last claim follows by straightforward expansion of $
  f_H(x)-f_{H}^*(x)$ about $x=0$. Details are given in the appendix.
\end{proof} 

The additive closeness of the spectrum is a highly non-trivial result:
from the usual spectrum definition of LRD (Section \ref{ssec:lrdss}),
LRD with Hurst parameter $H$ implies only that the ratio between
$f_H/f_H^*$ is bounded at the origin whereas the difference
$f_H-f_H^*$ generally diverges. That the difference is not only a
bounded function but tends to zero, and is also differentiable,
emphasizes in yet another way how unusual fractionally differenced
processes are among LRD processes. To explore this in more detail, observe that 
the statement of Corollary~\ref{cor:spectral-closeness} can be written
\begin{displaymath}
         f_H + \varphi^- = f_H^* + \varphi^+
\end{displaymath}
where $\varphi^-\ge0$ and $\varphi^+\ge0$.
Both $\varphi^+$ and $\varphi^-$ define spectral densities with
$\varphi^+(0)=\varphi^-(0)=0$. We then (\cite{tstm}, Cor.~4.3.1)
obtain a probabilistic variant of the closeness result: a fractionally
differenced process is equal in the distributional sense to its
limiting fGn up to additive independent processes with spectra
$\varphi^+,\varphi^-$, both of which have the property of having a
vanishing covariance sum $S_\g=\sumj \g_j$.  Such processes (called
\textit{Constrained Short Range Dependent} (CSRD) in \cite{discrete}),
lie in the DoA of an fGn with Hurst parameter $H'\in[0,1/2)$. 
 In contrast, for Short Range Dependent (SRD)
processes (those in the DoA of a fGn with $H'=1/2$), $S_\g$ is finite
but positive.  A graph of a particular $\varphi$ and its first
derivative is shown in Figure~\ref{fig:phi}. The plot suggests that
$\varphi^- \equiv 0$; whereby FARIMA would be equal in distribution to fGn
plus an independent CSRD process.
\begin{figure}
  \centering
  \includegraphics[scale=\oneup,angle=-90]{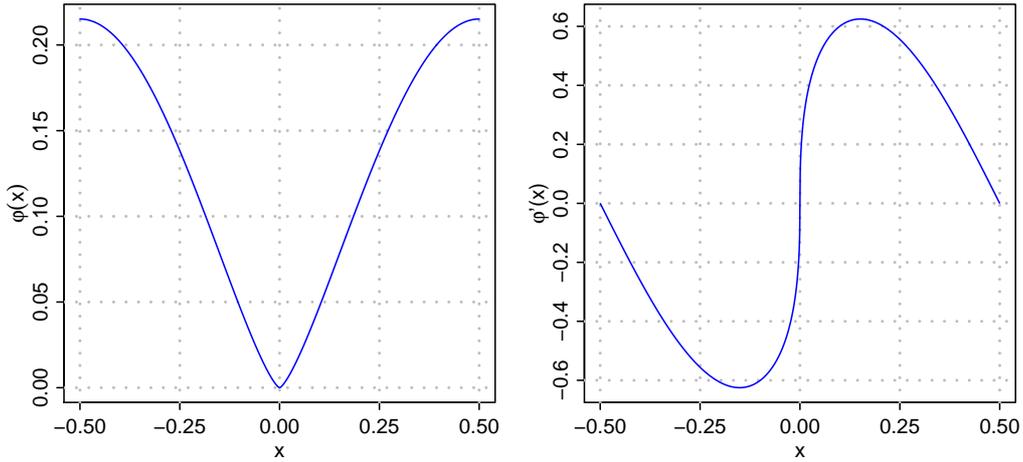}  
  \caption{\label{fig:phi} The function $\varphi(x)=f_H(x) - f^*_{H}(x)$ and its first derivative in the canonical case of a
    pure fractionally differenced process (FARIMA0d0) with $H=0.8$ and $\cffar=\cffgn$.}
\vspace{-3mm}
\end{figure}

To conclude our treatment of the ACVF, observe that a slightly weaker form of the closeness result of
Theorem~\ref{thm:conv-of-farima-vtf-2} can be derived from 
Theorem~\ref{prop:conv-of-farima}. 
Indeed, the identity $\omega_H(n)-\omega_{H}^*(n)=(\I d)_n$ implies
\begin{displaymath}
  |\omega_H(n)-\omega_{H}^*(n)|
  =  \Big|\sum_{k=0}^{n-1} \Big(\sumj d_j - \sum_{j=-k}^k d_j\Big) \Big| \leq 2 \sum_{k=0}^{n-1}\sum_{j=k+1}^\infty
  |d_j| \leq O(1) \sum_{k=0}^{n-1}k^{2H-3} = O(1),
\end{displaymath}
where we have used that $d_n=O(|n|^{2H-4})$ implies
$\sum_{j=k+1}^\infty |d_j|=O(1) \sum_{j=k+1}^\infty j^{2H-4}=O(k^{2H-3})$. 
The  $O(1)$ remainder term simply corresponds to a bounded function; this is clearly somewhat weaker than the
asymptotically constant remainder term appearing in Theorem~\ref{thm:conv-of-farima-vtf-2}.

We recently became aware of Lieberman \& Phillips (2008) \cite{LP2008} which provides an
asymptotic expansion for a class of fractionally differenced processes
corresponding to \eref{eq:farimaeq}, though $h(x)$ is required to be smooth rather than $C^3$.  
Using the first two terms of this expansion and comparing with an expansion for $\gstar(m)$, it is possible to
recover the $O(n^{2H-4})$ term of Theorem~\ref{prop:conv-of-farima}.
The work of \cite{LP2008} is focussed on numerical approximation
through infinite-order asymptotic expansions and does not compare
against fGn or draw conclusions on convergence speed or brittleness as we do here.

\section{Fractional Processes are Brittle}
\label{sec:brittle}

As pointed out at the end of Section~\ref{sec:VTF}, fractionally
differenced processes converge `almost immediately' to their fGn fixed
point compared to other processes in the domain of attraction, and
this is true in terms of each of the VTF, ACVF and spectrum.  In this
section we point out and illustrate a key consequence of this fact,
namely the \textit{brittleness} of fractionally differenced models.

\subsection{Brittleness}
\label{ssec:brittle}

Experimental data, especially data measured on a continuous scale, is very rarely
clean.  Imperfections in physical measurement are often treated
through the concept of observation noise, modelled as a random process
which perturbs the underlying observables.  A very common choice is
that of additive independent Gaussian noise, either white or coloured.
In the present context, this corresponds to adding to the original VTF
(or ACVF, or spectrum) the VTF (respectively ACVF, spectrum) of a
short range dependent noise process, that is a noise whose own fGn fixed point has $H'=1/2$.

As argued at the end of the previous section, we can essentially think of a
fractionally differenced process as an fGn to which a CSRD
process has been added.  Adding an SRD noise to this will change the
asymptotic behaviour, because the SRD asymptotics (with $S_\g>0$) is
`stronger' than CSRD asymptotics (with $S_\g=0$).  In terms of the
hierarchy within the DoA described by the index $\beta$ from
\eref{eq:RV}, whereas the original process lies very close to the
centre with $\beta=0$, the SRD-perturbed process will lie considerably
further out, with $\beta=1$.  A similar observation can be made if we
instead add a noise with LRD with $H'<H$ (resulting in
$\beta\in(1,2H)$), or even another CSRD process with $H'>0$ (resulting
in $\beta\in(0,1)$).
This last result follows from the fact that Theorem~\ref{thm:conv-of-farima-vtf-2} implies that
the `error' processes are so special that they are not only CSRD, but correspond to the extreme
case of $H'=0$, resulting in $\beta=0$.

Since the addition of even trace amounts of noise of diverse kinds will change the asymptotics, pushing the process further from its fGn limit and therefore slowing its convergence rate to it under aggregration, fractional differencing models are `brittle' or non-robust in this sense.  
Properties of systems driven by such processes may therefore differ qualitatively from properties of the same system once noise is added. 
The precise impact of the noise is beyond the scope of this paper (see the discussion).  It will depend on both the application and the class of noise and must be determined case by case.

\subsection{Numerical Illustrations}
\label{ssec:num}

In this section we illustrate the brittle nature of fractionally differenced processes through high accuracy numerical evaluation of the VTF of FARIMA time series, both with and without additive noise.

Three  different examples will be considered, two with SRD-noise and one with LRD-noise. 
More precisely,  the perturbed processes are $Z_i(t) = X_i(t)+\sqrt{0.1}Y_i(t)$ for $i=1,2,3$, where
\begin{enumerate}
\item[\textbf{1)}] $X_1$:  unit variance $\mathrm{FARIMA}(0,0.3,0)$; \\
                        \  $Y_1$:  unit variance Gaussian white noise,
\item[\textbf{2)}] $X_2$: unit variance $\mathrm{FARIMA}(1,0.3,1)$ with ARMA parameters $(\phi_1,\theta_1)=(0.3,0.7)$;\\  
                         \, $Y_2$: unit variance $\mathrm{ARMA}(1,1)$ process also with ARMA parameters $(\phi_1,\theta_1)=(0.3,0.7)$,
\item[\textbf{3)}] $X_3$: unit variance $\mathrm{FARIMA}(0,0.3,0)$;\\ 
                        \, $Y_3$: unit variance  $\mathrm{FARIMA}(0,0.2,0)$.
\end{enumerate}
\begin{figure}[!tb]
  \centering
  \includegraphics[scale=\oneup,angle=-90]{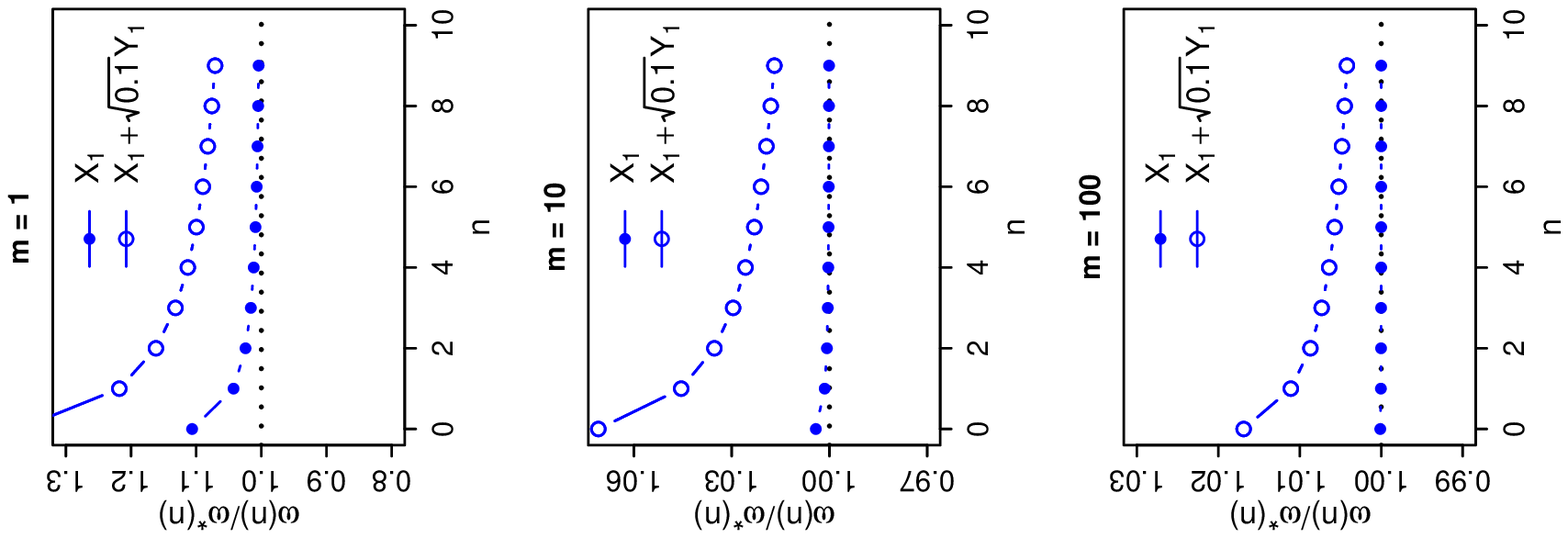}  
  \includegraphics[scale=\oneup,angle=-90]{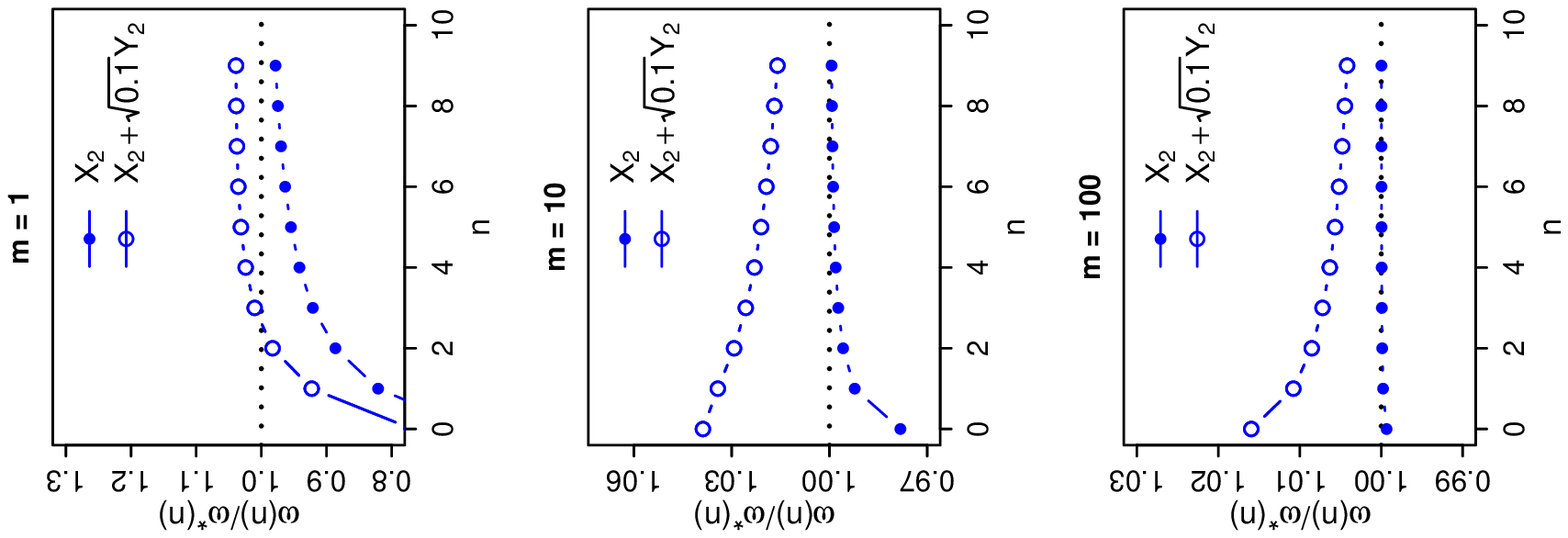}  
  \includegraphics[scale=\oneup,angle=-90]{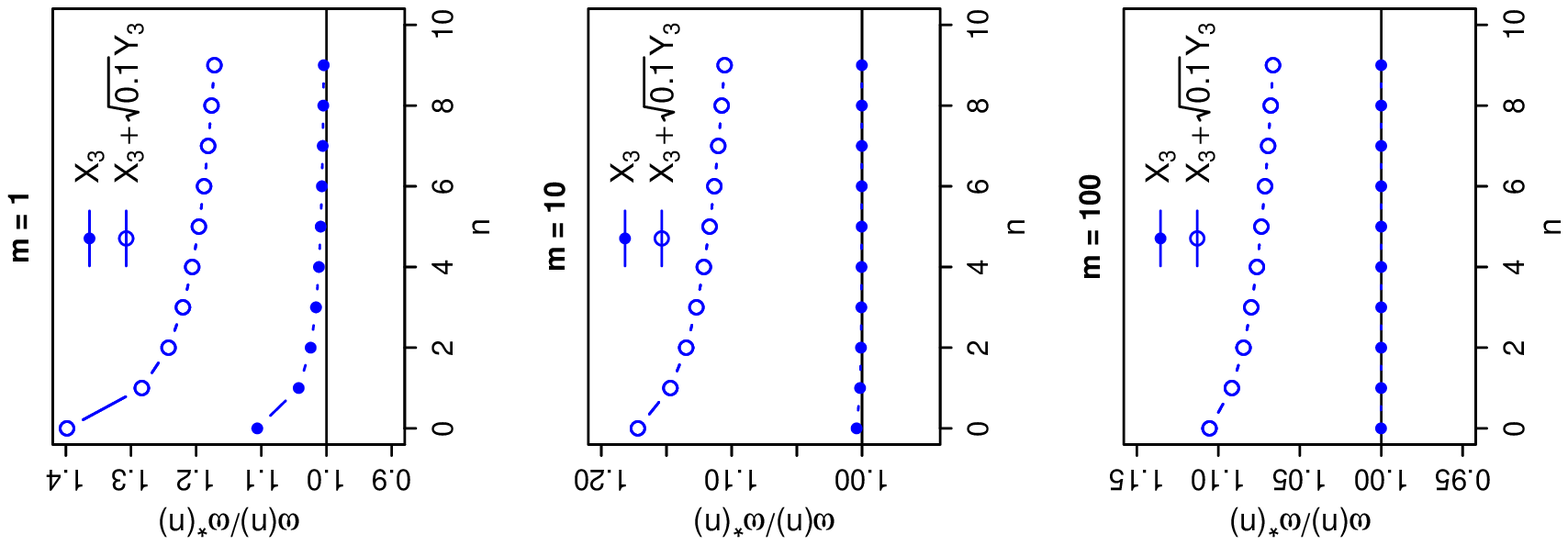}
 \caption{\label{fig:numerics} {Ratios of VTF's of original FARIMA and perturbed processes to their fGn limit, both originally and under aggregation, one column per example.
 The solid circles denote unperturbed FARIMA; the hollow circles the perturbed ones. 
 It is seen that the VTF's for unperturbed FARIMA converge much faster} than their perturbed counterparts.}
 \vspace{0mm}
\end{figure}

In each case, the original process $X_i$ and the perturbed process
$Z_i$ share a common fGn fixed point, but have unequal variances.  It
may seem unfair to compare results for processes with different
variances, however the opposite is true.  In fact, if the variances of
$Z_i$ and $X_i$ were chosen equal, this would mean that $\cf\ne\cf^*$,
and so their fGn limits would be different, rendering meaningful comparison
impossible.  To see this more directly, from the definitions in
Section~\ref{ssec:lrdss} it is clear that adding a perturbation
corresponding to a smaller $H$ value does not alter the fixed
point. On the other hand the variance must increase when an
independent noise is added.

For each example $i=1,2,3$, we calculate the VTF of $Z_i$ and $X_i$ and
normalise them by dividing by their common fGn limit $\w^*_H$.
Closeness to fGn can therefore be evaluated by looking to see how the
normalised VTF deviates from 1 for each lag.  Maple version 13 was
used to numerically evaluate the variance time functions to a high
degree of precision.

Figure~\ref{fig:numerics} displays the normalised VTFs for lags 1-10
for aggregation levels $m=1$, $10$, and $100$, with one example per
column. The graphs clearly demonstrate that even a small departure
from FARIMA takes the process much further away from its corresponding
fGn.  Indeed, after an aggregation of level 100, in each case the VTF
of the original process is visually indistinguishable from its fGn
limits compared to their perturbed versions.

Note that the second column in the figure gives an
example where before aggregation ($m=1$) the perturbed process was in
fact \textit{closer} to the fixed point over the first few lags, where
most of the obvious autocovariance lies. Under aggregation however,
this quickly reverses as the different asymptotic behaviours of the
original and perturbed processes manifest and become dominant at all lags.

\section{Discussion}
\label{sec:conc}

We have shown that fractionally differenced processes have an
asymptotic autocovariance structure which is extremely close to that
of the fractional Gaussian noise, more specifically, to that of the
fGn fixed point to which the given process will tend under aggregation based
renormalisation.  We have shown this independently for each of three
equivalent views of the autocovariance structure, namely behaviour of
the spectral density at the origin, and each of the autocovariance
function and the variance time function in the large lag limit.

We showed that the natural class of processes against which this
behaviour should be compared are those in the domain of attraction of
the fGn fixed point limit.  Using regular variation to provide a
measure of distance from this fixed point within the DoA, we were able
to precisely quantify the nature of this `closeness', and to confirm
that the fractionally differenced class are indeed exceptionally
unusual in this regard,  resulting in very fast convergence to fGn
  under renomalisation.  We then used this fact to point out that the
fractionally differenced process class is brittle, that is, non robust
to the presence of noise.  In particular we showed that the addition
of arbitrarily small amounts of independent noise,  not only
  Gaussian white noise but also noises which are much gentler in a
  precise sense, changes the asymptotic covariance structure
qualitatively.  This fact has not been appreciated in the literature
where such models, for example the FARIMA class, are widely used in
time series modelling, synthetic data generation, and to drive more
complex stochastic systems such as queuing systems, without regard to
robustness with respect to the model in this sense.

The assessment of the impact of the brittleness of fractionally
differenced models is beyond the scope of this work, as it will depend
intimately on each particular application as well as the nature of the
noise in question.  However, we argue that conclusions based on the
perception that FARIMA and related models represent `typical' LRD
behaviour need to be reassessed, in particular in contexts where noise
is important to consider.  To give an example of a possible impact in
the noiseless case, we conclude by expanding upon the comments given
in the introduction on statistical estimation.
    
The closeness of a process to its fGn fixed point in functional terms
is directly related to the speed of convergence of that process to the
fixed point under aggregation.  One application where this fact
carries direct implications is the performance of statistical
estimators for the Hurst parameter $H$.  Fundamentally,
semi-parametric estimators of scaling parameters such as $H$ are based
on underlying estimates made at a set of `aggregations' at different
levels, that is at multiple scales
\cite{Robinson94,Beran,A2,Comparestimate95}.  The sophistication of
particular estimators notwithstanding, this is true regardless of
whether they are based in the spectral, time, or wavelet domains,
though the technical details vary considerably.  In the time domain
using time domain aggregation the link is of course direct, and
reduces to looking at the asymptotically power-law nature of
$\vxm=\m\w(0)$ as a function of $m$ in some form.  This is precisely
where fractionally differenced processes are at a real advantage, as
this quantity converges extremely quickly to that of the fGn fixed
point, whose ideal power-law behaviour $\vxm=\vx m^{2H}$ allows $H$ to
be easily recovered.  As a result, estimator performance evaluated
through the use of fractionally differenced models would be superior
to that for LRD processes more generally.
Note that we are not recommending that $H$ estimation be performed
directly in the time domain by regressing $\m{\hat{\vx}}$ on $m$,
indeed we have argued the opposite \cite{A2}.  Our point is that the
extreme closeness of such models to fGn must ultimately manifest in
simpler asymptotic behaviour which will, in general, translate to
improved estimation.  Indeed, in the spectral domain, the importance
of the degree of smoothness at the origin for the ultimate limits on
estimator performance has already been noted \cite{GRS1997}.  Note
that the above observations in no way put into question the findings
of prior work on estimation of fractional processes in noise.

\section{Appendix}
\label{sec:app}

The appendix is split according to results relating to spectral closeness (Section~\ref{ssec:closespec}), 
closeness of the VTF (Section~\ref{ssec:closeVTF}), and of the ACVF (Section~\ref{sec:ACVF}).
For convenience, the statement of results proved here are generally repeated.  
Lemmas A1 and A2 are labelled separately as they appear in the appendix only.

\setcounter{lemma}{0}

\subsection{Spectrum}
\label{ssec:Aspec}

\noindent\textbf{Details of the proof of Theorem~\ref{lem:spec-dens-ratio}}

\textbf{(i)} \ It is well known, and can be verified by examining
(\ref{eqn:ffgn}) and (\ref{eqn:ffar}), that each of $f_H^*(x)$ and
$f_H(x)$ diverge to infinity at $x=0$ but are otherwise even, positive
and continuous.  Since $g(0)>0$ is finite, $g$ is positive and
continuous on the compact domain and hence bounded, and even.  Since
$g$ is continuous, it is integrable (\cite{Champeney90}, p.9), and
since, for $p>0$, $g^p$ is likewise continuous, positive and bounded,
$g$ is in $L^p$.

\textbf{(ii)} \  Since $a(x)>0$, $\tilde g$ is bounded away from zero.
Each of $a$, $b$, and $c$ are smooth.  The latter follows from the fact that for each $j\ne0$ the term $|\pi j +  \pi x|^{-(2H+1)}$ is infinitely differentiable in $x\in[-1/2,1/2]$. 
By comparing against $\sum_{j=1}^\infty (j-1/2)^{-(2H+1)}<\infty$ the Weierstrass' $M$-test shows 
that the defining sum for $c$, and the sum of the term by term first derivatives, each converge uniformly.
A classical result on the differentiability of infinite series (\cite{Apostol74}, Thm.~9.14) then shows 
that $c'$ is given by the latter sum.  Using exactly the same $M$-test, this can be 
repeated for derivatives of all orders, proving that $c$ is smooth.

Since $a$ is smooth and bounded above zero, $|a|^{2H+1}$ is smooth over $[-1/2,1/2]$, 
and the same is true for  $|b|^{2H+1}$ away from the origin. 
It follows that $\tilde{g}$ is smooth everywhere except at the origin where its smoothness is
controlled by that of $|b|^{2H+1}$.




\subsection{VTF}
\label{ssec:Avtf}

\begin{lemma}
    The auto-covariance functions $\gamma_H$ and $\gstar$  are related through the convolution $\gamma_H= G\star \gamma_{H}^*$.
\end{lemma}
\begin{proof}
\vspace{-2mm}
The r.h.s.~exists since   $\sumj G_j \gamma_N^*(n-j)\le
  \sumj |G_j| |\gamma_N^*(n-j)|\le \gamma_N^*(0) \sumj |G_j|<\infty$ from Theorem~\ref{lem:spec-dens-ratio}.  
For the l.h.s.~we can write
 \begin{eqnarray}
   \gamma_H(n) &=& \inth f_{H}(x)       \,e^{2\pi ix n} \mathrm{d}x 
                =  \inth g(x)f_{H}^*(x) \,e^{2\pi ix n} \mathrm{d}x \\
               &=& \inth \bl\sumj G_je^{-2\pi ix j}  \br  f_{H}^*(x) \,e^{2\pi ix n} \,\mathrm{d}x 
   \label{eqn:iteratedint}
 \end{eqnarray}
 as the Fourier series for $g(x)$ converges absolutely for all $x$ since
 $\sumj|G_j|<\infty$ (Theorem~\ref{lem:spec-dens-ratio}). 
Now
\begin{eqnarray*}
  \inth \bl\sumj |G_je^{-2\pi ix j}|  \br  f_{H}^*(x) \,e^{2\pi ix n} \,\mathrm{d}x &=&
            \sumj |G_j| \inth  f_{H}^*(x) \,e^{2\pi ix n} \,\mathrm{d}x  \\
            &=&\gamma_{H}^*(n)  \sumj |G_j| \  \le \infty \, .
\end{eqnarray*}
This justifies the use of Fubini's Theorem (\cite{SJTaylor}, Th.6.5)
on the iterated integral (\ref{eqn:iteratedint}) to reverse the order of
integration and summation. Using the evenness of $G$ and $f_H^*$, this yields
\vspace{-1mm}
 \begin{eqnarray*}
    \gamma_H(n)           
      &=& \sumj G_j \inth f_{H}^*(x) \cos(2\pi xj) \cos(2\pi xn) \,\mathrm{d}x \\
      &=& \sumj G_j \inth f_{H}^*(x) \frac{1}{2}\bl\cos(2\pi x(j-n)) + \cos(2\pi x(j+n)\br \,\mathrm{d}x \\
      &=& \frac{1}{2}\sumj G_j \big(\gamma^*_{H}(j-n)+\gamma^*_{H}(j+n)\big) =
          \frac{1}{2}\bl \sumj G_j \gamma^*_{H}(n-j) +  \sumj G_j \gamma^*_{H}(-j-n) \br \\
      &=& \frac{1}{2}\bl (G\star\gamma^*_H)(n) +  (G\star\gamma^*_H)(-n) \br = (G\star\gamma^*_H)(n),
 \end{eqnarray*}
using the evenness of $\gamma_H^*$ and $G\star \gamma_{H}^*$, and the
existence of $G\star\gamma_{H}^*$ to justify the splitting of the sum.
\end{proof}

\bigskip
\begin{lemma}
Assume $1 <\alpha < 2$ and let $a=\{|n|^{\alpha}\,:\,n \in \Z \}$. 
Let $b$ be a symmetric sequence satisfying $\sum_{j=1}^\infty j^\alpha |b_j|<\infty$. 
Then $S_b=\sum_{j=-\infty}^\infty b_j$ and the symmetric sequence $c=a \star b$
exist, and $(c_n - S_b a_n)\ngoinf 0$.
\end{lemma}
\begin{proof} Since $\alpha>1$, $\sum_{j=-\infty}^\infty |b_j| \le
|b_0|+2\sum_{j=1}^\infty j^\alpha|b_j|<\infty$, so $b$  is absolutely summable and hence summable.
Now consider $c$. Clearly $c_0=\sum_{j=-\infty}^\infty |-j|^\alpha b_{j}$
exists by the assumptions on $b$, and for $n>0$
\begin{eqnarray*}
  \label{eq:ex-of-conv}
  |c_n| = |(a \star b)_n|
    &\leq&  \sum_{j=-\infty}^{-n}|n-j|^\alpha |b_j| + \sum_{j=-n+1}^{n-1}|n-j|^\alpha|b_j|
                                                    + \sum_{j=n}^{\infty}|n-j|^\alpha|b_j|\\
    &\leq&  \sum_{j=n}^{\infty}(2j)^\alpha |b_j|    + \sum_{j=-n+1}^{n-1}|n-j|^\alpha|b_j|
                                                    + \sum_{j=n}^{\infty}j^\alpha|b_j|<\infty .
\end{eqnarray*}
Since both $a$ and $b$ are symmetric, $c_n$ also exists for $n<0$, and so $c$ exists and is symmetric.\\
For the last part, since $c_n - S_b a_n$ is symmetric in $n$ we assume $n\ge0$ and rewrite it as
\begin{displaymath}
   \sum_{j=-\infty}^\infty |n-j|^\alpha b_{j}  - n^\alpha \sum_{j=-\infty}^\infty b_j 
        = n^\alpha b_0 + \sum_{j=1}^\infty (n+j)^\alpha b_{j} + 
          \sum_{j=1}^\infty |n-j|^\alpha b_{j} - n^\alpha \sum_{j=-\infty}^\infty b_j
        = \sum_{j=1}^\infty T_{n}^j b_j
\end{displaymath}
where $T_n^j:=|n-j|^\alpha+(n+j)^\alpha-2n^\alpha$, $n\ge0$, $j>0$. 
Noticing that $T_n^j= f_\alpha(n,j)$ from 
Lemma~A1, we
have that $T_n^j < T_0^j=2j^\alpha$ for each fixed $j$, and so
\begin{displaymath}
  |c_n - S_b a_n| \leq \sum_{j=1}^{N}  |T_n^j||b_j| +    \sum_{j=N+1}^\infty  |T_n^j||b_j|  
                  <    \sum_{j=1}^{N}  |T_n^j||b_j| + 2\!\sum_{j=N+1}^\infty j^\alpha|b_j| .
\end{displaymath}

Now given any $\varepsilon>0$, a $N(\varepsilon)>1$ can be found such that
$\sum_{j=N+1}^\infty j^\alpha|b_j|<\varepsilon/4$.  
Next, since $T_n^j \ngoinf 0$ for any fixed $j$ 
(Lemma~A1 below),
there exists an $n_0(N)$ such that
$\sum_{j=1}^{N} |T_n^j||b_j|<\varepsilon/2$ when $n \geq n_0$. 
It follows that $|c_n - S_b a_n|<\varepsilon$ for $n \geq n_0$ and so
 $(c_n - S_b a_n)\ngoinf 0$.
\end{proof}

\medskip
{\noindent\textbf{Lemma A1}\textit{
\label{lem:dom-fun-2}
Assume  $1< \alpha< 2$ and define
$f_\alpha(x,y):=|x-y|^{\alpha}+(x+y)^\alpha-2x^\alpha$ for  $x\ge0$,  $y>0$. 
For each $y$, $f_\alpha(\,\cdot\,,y)$ is positive, strictly decreasing, and
$\lim_{x \to \infty} f_\alpha(x,y)=0$.
\vspace{-3mm}
}
}
\begin{proof}
\quad  Fix $y>0$.  We split the domain of $f_\alpha(\,\cdot\,,y)$ into two cases.\\
Let $x\geq y$.  It follows that $f_\alpha'(\,\cdot\,,y)=\alpha f_{\alpha-1}(\,\cdot\,,y)$. 
Define $g(x) = x^{\alpha}$. Since $g'(x)=\alpha x^{\alpha-1}$ is strictly concave,  
$(x-y)^{\alpha-1}+(x+y)^{\alpha-1} < 2x^{\alpha-1}$ and so
$f_\alpha'(\,\cdot\,,y) <0$ and $f_\alpha(\,\cdot\,,y)$ is strictly decreasing.  
To prove  $\lim_{x \to \infty} f_\alpha(x,y)=0$, we apply the mean
value theorem twice to $g$, and then once to $g'$, to obtain:
\begin{eqnarray}
  f_\alpha(x,y) &=&  \big((x+y)^\alpha-x^\alpha\big)  -  \big(x^\alpha-(x-y)^\alpha\big)\\
                &<& \alpha y\big((x+y)^{\alpha-1}-(x-y)^{\alpha-1}\big) \\
                &<& 2\alpha(\alpha-1)y^2(x-y)^{\alpha-2}
\end{eqnarray}
(since $g'$ is strictly increasing and $g''$ strictly decreasing), which tends to zero
as $x\to\infty$.\\
\smallskip
Let $x<y$.  In this case, the derivative with respect to $x$ yields
\ba
 f'_\alpha(x,y) &=& \alpha\bigl( (x+y)^{\alpha-1}-(y-x)^{\alpha-1}- 2x^{\alpha-1} \bigr)\\
                &<& \alpha\bigl( (x+y)^{\alpha-1}-(y-x)^{\alpha-1}- (2x)^{\alpha-1}\bigr) \\
                &=& \alpha\bigl( h_x(y) - h_x(x)\bigr)
\ea
where $h_x(y)= (x+y)^{\alpha-1}-(y-x)^{\alpha-1}$. Since the
derivative of $h_x$ is negative for $x>0$, $h_x$ is strictly decreasing. 
It follows that  $f'_\alpha(\cdot,y)<0$ and so  $f_\alpha(\,\cdot\,,y)$ is likewise strictly decreasing.\\
Finally, since $f_\alpha(x,y)$ is decreasing for all $x\ge0$ and tends to zero, it is positive.
\end{proof}

\bigskip
\begin{lemma}
  Let $a,b$ be symmetric sequences and assume that $c:=a \star b$
  exists.  Then $\I c$ exists, and if $(\I a) \star b$ exists, then $\I c =
  (\I a) \star b - \left((\I a) \star b\right)_{0}$ .
\vspace{-2mm}
\end{lemma}
\begin{proof}
Since $(\I c)_n$ is a finite sum of elements of $c$, it exists for each $n$.
 Now 
 \begin{displaymath}
   (\I c)_n= \sum_{k=0}^{n-1} \sum_{i=-k}^k\sum_{j=-\infty}^\infty a_j b_{i-j}
 \end{displaymath} 
 can be rewritten as $(\I c)_n = \sum_{j=-\infty}^\infty a_j H_n(j)$ where
$H_n(j):=\sum_{k=0}^{n-1} \sum_{i=-k}^k b_{i-j}$, since a finite sum of convergent series is convergent.
Since $(\I a)_{j-1} - 2(\I a)_j + (\I a)_{j+1}= a_{-j} + a_j = 2a_j$, we have
  \begin{eqnarray}
  \hspace{-6mm}  (\I c)_n = \sum_{j=-\infty}^\infty a_j H_n(j)
    &=& \frac{1}{2}\sum_{j=-\infty}^\infty\big((\I a)_{j-1} -
                           2(\I a)_j+(\I a)_{j+1}\big)H_n(j) \\
    &=& \frac{1}{2}\!\left(\sum_{j=-\infty}^\infty\!\!(\I a)_{j-1} H_n(j) -
                        2\!\!\!\sum_{j=-\infty}^\infty(\I a)_j \, H_n(j) + 
                         \!\!\!\sum_{j=-\infty}^\infty\!(\I a)_{j+1} H_n(j)\!\right)
  \label{eq:int-dec}\\
    &=&\frac{1}{2}\sum_{j=-\infty}^\infty  (\I a)_j\big(H_n(j+1)-2H_n(j)+H_n(j-1)\big),
  \label{eq:int-decc}
\end{eqnarray}
Step \eqref{eq:int-dec} is justified since each of the sums is
convergent, because each can be written as a finite sum of series of the form
$\sum_{j=-\infty}^\infty (\I a)_j\, b_{m-j}$ for some $m$, and this is
just $((\I a) \star b)_m$ which exists by assumption.
Now
\begin{eqnarray}
  && \hspace{-7mm} H_n(j+1)-2H_n(j)+H_n(j-1) = \big(H_n(j-1)-H_n(j)\big)-\big(H_n(j)-H_n(j+1)\big)\\
  &=& \sum_{k=0}^{n-1}\left(\Big(\sum_{i=-k-j+1}^{k-j+1}  b_i-\sum_{i=-k-j}^{k-j} b_i\Big) 
                          -\Big(\sum_{i=-k-j}^{k-j}     b_i-\sum_{i=-k-j-1}^{k-j-1} b_i\Big)\right)\\
  &=& \sum_{k=0}^{n-1}\Bigl(\big(b_{k-j+1}-b_{-k-j}\big) - \big(b_{k-j}-b_{-k-j-1}\big)\Bigr)\\
  &=& \sum_{k=0}^{n-1}\big(b_{k-j+1}-b_{k-j})-\sum_{k=0}^{n-1}\big(b_{-k-j}-b_{-k-j-1}\big)\\
  &=&  (b_{n-j} -b_{-j}) - (b_{-j} - b_{-n-j})  =  b_{n-j} + b_{-n-j} - 2b_{-j}.
\end{eqnarray}
The result then follows by substitution into \eqref{eq:int-decc}, using the existence of
$(\I a) \star b$ to justify splitting the sum, and finally by the symmetry of $(\I a)$ and $b$.
\end{proof}

\subsection{ACVF}
\label{ssec:AACVF}

\begin{lemma}
Assume $-1 \leq \alpha< 0$ and let $a$ be the symmetric positive sequence $a_n=|n|^\alpha$, $n \neq 0$ and $a_0>0$. 
Let $b$ be a symmetric sequence with $|b_0|<\infty$ for which there exists  $\beta \in [0,2]$ such that 
$\sum_{j=1}^\infty j^\beta|b_j|<\infty$ and $|b_n| = O(n^{-(\beta+1)})$. 
Then $S_b:=\sum_{j=-\infty}^\infty b_j$ and the symmetric sequence $c:=a\star b$ exist, 
and $c_n- S_ba_n=O(n^{\alpha-\beta})$ as $n \to \infty$.
\end{lemma}
\begin{proof}
  We have $\sum_{j=-\infty}^\infty |b_j|\leq b_0+2\sum_{j=1}^\infty
  j^\beta|b_j|<\infty$ so $b$ is absolutely summable and therefore summable. Then $S_b$ exists. 
  Moreover, $|c_n| = |(a \star b)_n| \leq \sum_{j=-\infty}^\infty |b_j| |a_{n-j}| \le |b_n|a_0 + \sum_{j=-\infty}^\infty |b_j|<\infty$. 
We conclude that $c_n$ exists for each $n \in \Z$, and that $c$ is symmetric by the symmetry of $a$ and $b$.  
Define $T_n^j:=a_{|n-j|}+a_{n+j}-2a_{n}$, and using the symmetry of $a$ and $b$ rewrite $c_n-S_ba_n$ as:
\begin{equation}
    (a \star b)_n - S_ba_n = a_n b_0 + \sum_{j=1}^\infty a_{|n-j|}b_{j} + \sum_{j=1}^\infty a_{n+j} b_{j}-  S_ba_n =\sum_{j=1}^\infty T_{n}^j b_j .
\end{equation}
To prove the last part of the theorem it suffices to consider $n\ge0$, since $c$ is symmetric, and 
as we are interested in large $n$ asymptotics, we restrict to $n>2$.
The sum for $|c_n-S_b a_n|$ can be decomposed as
\begin{equation}
  |c_n-S_b a_n| \leq \sum_{j=1}^{\lfloor{n/2}\rfloor}       |T_n^j||b_j| + \!
                \sum_{j=\lfloor{n/2}\rfloor+1}^{2n} |T_n^j||b_j| \,+
                \sum_{j=2n+1}^\infty                        |T_n^j||b_j| \,=:  A_n + B_n + C_n .
\end{equation}
We shall show that each of $A_n,B_n$, and $C_n$ are of order $O(n^{\alpha-\beta})$.

The definition of $A_n$ implies $n>j>0$, so 
Lemma~A2 below
applies to $T_n^j=f_\alpha(n,j)$, and implies the
existence of a constant $K>0$ such that $|T_n^j| \le Kj^2(n-j)^{\alpha-2} < K j^2(n/2)^{\alpha-2}$ when $j\le n/2$. Thus
\begin{equation}
    K^{-1}2^{\alpha-2}A_n \leq n^{\alpha-2}                 \sum_{j=1}^{\lfloor{n/2}\rfloor} j^{2-\beta}j^\beta|b_j|
                                   \leq n^{\alpha-2} n^{2-\beta} \sum_{j=1}^{\lfloor{n/2}\rfloor} j^\beta|b_j| 
                                   \leq n^{\alpha-\beta} \!\sum_{j=1}^\infty j^\beta|b_j| = O(n^{\alpha-\beta}) .
\end{equation}

For $B_n$, where $n\neq j$ and $n,j>0$, we have $|T_n^j| < 2|n-j|^{\alpha}$, while $|T_n^n| =a_0+(2^\alpha-2)a_n=O(1)$. 
Then for sufficiently large $n$, by assumption there exists a $K>0$ such that
\begin{eqnarray*}
\label{eq:ajk}
  B_n   &\leq& 2\!\sum_{j=\lfloor n/2\rfloor +1}^{n-1} \!(n-j)^{\alpha}|b_j| + 2\!\sum_{j=n+1}^{2n} (j-n)^{\alpha} |b_j| \,+ |T_n^n||b_n| \\ 
           &<& 2K \bl\frac{n}{2}\br^{\!-(\beta+1)}\left(\sum_{j=\lfloor n/2\rfloor+1}^{n-1} (n-j)^{\alpha}+\sum_{j=n+1}^{2n} (j-n)^{\alpha} + |T_n^n|  /2\right) \\
           &<& 2^{\beta+3}K n^{-(\beta+1)}\sum_{j=1}^{n} j^\alpha                                                   + O(n^{-(\beta+1)})\\
           &<& 2^{\beta+3}K n^{-(\beta+1)} \Big( 1+\int_1^n x^\alpha \mathrm{d}x\Big) + O(n^{-(\beta+1)}) = O(n^{\alpha-\beta}) .
\end{eqnarray*}
For $C_n$, where $j \geq 2n$, we have $T_n^j < 2n^{\alpha}$. Since $\sum_{j=2n+1}^{\infty} |b_j| \leq
(2n)^{-\beta}\sum_{j=2n+1}^{\infty} j^\beta|b_j|=o(n^{-\beta})$, we get
\begin{equation}
   C_n \leq \sum_{j=2n+1}^\infty  2n^{\alpha}|b_j| \leq 2n^{\alpha} \!\sum_{j=2n+1}^{\infty} |b_j| = o(n^{\alpha-\beta}). 
\end{equation}
Conclude that $|c_n-S_b a_n|=O(n^{\alpha-\beta})$ as $n \to \infty$.
\end{proof}

\smallskip
{\noindent\textbf{Lemma A2}\textit{
 \label{lem:dom-fun}
 Assume $\alpha < 0$ and define $f_\alpha(x,y):=|x-y|^{\alpha}+(x+y)^\alpha-2x^\alpha$ for $x>y>0$. 
 Then $f_\alpha(x,y)$ is positive and obeys $f_\alpha(x,y) <2\alpha(\alpha-1)y^2(x-y)^{\alpha-2}$.
}
}
\vspace{-3mm}
\begin{proof}
Since $x>y$ it follows that $f_\alpha'(\,\cdot\,,y)=\alpha f_{\alpha-1}(\,\cdot\,,y)$. 
Define $g(x) = x^{\alpha}$. Since $g'(x)=\alpha x^{\alpha-1}$ is strictly concave,  
$\alpha(x-y)^{\alpha-1}+\alpha(x+y)^{\alpha-1} < 2 \alpha x^{\alpha-1}$ and so
$f_\alpha'(\,\cdot\,,y) <0$ and $f_\alpha(\,\cdot\,,y)$ is strictly decreasing.  
Now apply the mean value theorem twice to $g$, and then once to $g'$, to obtain:
  \begin{eqnarray*}
    f_\alpha(x,y)   &=& \big((x+y)^\alpha-x^\alpha\big)-\big(x^\alpha-(x-y)^\alpha\big)\\
                             &<& \alpha y((x-y)^{\alpha-1}-(x+y)^{\alpha-1}) \\
                             &<& 2\alpha(\alpha-1)y^2(x-y)^{\alpha-2}.
  \end{eqnarray*}
since $g'$ is strictly increasing and $g''$ strictly decreasing.
\end{proof}

\noindent\textbf{Details of the proof of Corollary~\ref{cor:spectral-closeness}} \
We explain here why $\varphi(x)=f_H(x)-f_H^*(x)=O(x^{-2H+3})$ as $x\to 0$. 
Calculate the first few derivatives of the analytic function
$x \mapsto (\sin(x)/x)^{-2H+1}$ (set to 1 at $x=0$) and expand in a Taylor series around
the origin to find that $ (\sin(x)/x)^{-2H+1}=1+O(x^2)$. 
It follows that $\sin(x)^{-2H+1}=x^{-2H+1}+O(x^{-2H+3})$ for $x \neq 0$.  
The function $h$ is assumed three times (continuously) differentiable. 
Symmetry implies $h'(0)=0$ so that by Taylors theorem, $h(x)=h(0)+O(x^2)$. Thus
\begin{displaymath}
   h(x)\sin(\pi x)^{-2H+1}=h(0)\pi^{-2H+1} x^{-2H+1}+O(x^{-2H+3}),
   \quad x \neq 0,
\end{displaymath}
while it can be shown that
 \begin{displaymath}
 (1-\cos(2\pi x))x^{-2H-1}=2\pi^2x^{-2H+1}+O(x^{-2H+3}), \quad x \neq 0,
 \end{displaymath}
and
\begin{displaymath}
  (1-\cos(2\pi x))\sum_{\substack{j=-\infty \\ j \neq 0}}^\infty|\pi
  j+\pi x|^{-2H-1}=O(x^2).
\end{displaymath}
Then
\begin{eqnarray*}
  f_H(x)-f_{H}^*(x) 
&=&\big\{h(x)(2\sin(\pi x))^{-2H+1}\big\}\\
&& \phantom{XXX}-\big\{h(0)2^{-2H}\pi^{-2H-1}(1-\cos
(2\pi x))
  x^{-2H-1}\big\}+O(x^2) \\
  &=&\big\{h(0)\pi^{-2H+1} 2^{-2H+1}x^{-2H+1}+O(x^{-2H+3})\big\}\\
  && \phantom{XXX}-\big\{h(0)2^{-2H+1}\pi^{-2H+1}x^{-2H+1}+O(x^{-2H+3})\big\}+O(x^2)\\
  &=&O(x^{-2H+3}).
\end{eqnarray*}

\bibliographystyle{Chicago}

\end{document}